\documentclass[conference]{IEEEtran}

\usepackage{graphicx}
\usepackage{amsmath}
\usepackage{amsfonts}
\usepackage{amssymb}
\usepackage{mathrsfs}
\usepackage{eqparbox}
\usepackage{amsthm}
\usepackage{algorithm}
\usepackage{algorithmic}
\usepackage{verbatim}
\usepackage{cite}

\newtheoremstyle{bfnote}
  {}{}
  {\itshape}{}
  {\bfseries}{.}
  { }{\thmname{#1}\thmnumber{ #2}\thmnote{ (#3)}}
\theoremstyle{bfnote}
\newtheorem{theorem}{Theorem}
\newtheorem{lemma}{Lemma}

\newtheorem{definition}{Definition}

\begin{document}
\bibliographystyle{IEEEtran}
\title{\huge{Timely Synchronization with Sporadic Status Changes}}
\author{\IEEEauthorblockN{Chenghao Deng\IEEEauthorrefmark{1}, Jing Yang\IEEEauthorrefmark{2}, Changyong Pan\IEEEauthorrefmark{1}}\\
 \IEEEauthorblockA{\IEEEauthorrefmark{1}Department of Electronic Engineering, Tsinghua University\\ Beijing National Research Center for Information Science and Technology (BNRist), Beijing 100084, P. R. China\\ \IEEEauthorrefmark{2}School of Electrical Engineering and Computer Science, The Pennsylvania State University, University Park, PA 16802\\
 \textit{dengch16@mails.tsinghua.edu.cn, yangjing@psu.edu, pcy@tsinghua.edu.cn}}}

\maketitle


\begin{abstract}
In this paper, we consider a status updating system where the transmitter sends status updates of the signal it monitors to the destination through a rate-limited link.
We consider the scenario where the status of the monitored signal only changes at discrete time points.
The objective is to let the destination be synchronized with the source in a timely manner once a status change happens.
What complicates the problem is that the transmission takes multiple time slots due to the link-rate constraint.
Thus, the transmitter has to decide to switch or to skip a new update when the status of the monitored signal changes and it has not completed the transmission of the previous one yet.
We adopt a metric called ``Age of Synchronization'' (AoS) to measure the ``dissatisfaction'' of the destination when it is desynchronized with the source.
Then, the objective of this paper is to minimize the time-average AoS by designing optimal transmission policies for the transmitter.
We formulate the problem as a Markov decision process (MDP) and prove the multi-threshold structure of the optimal policy.
Based on that, we propose a low computational-complexity algorithm for the MDP value iteration.
We then evaluate the performance of the multi-threshold policy through simulations and compare it with two baseline policies and the AoI-optimal policy. 
\end{abstract}



\section{Introduction}
The ubiquitous network connectivity has enabled real-time status monitoring and control in various applications, such as smart home, autonomous driving, smart grids, etc.
In such applications, ensuring timely delivery of status updates to the controller is of critical importance for the stability, safety and efficiency of the system.
On the other hand, the underlying network infrastructure usually cannot support instantaneous delivery of the status update data.
It thus calls for universal metrics to measure the freshness of the status information available at the controller.

Recently, a few metrics have been introduced to measure information freshness.
Among them, the most prevalent one is \textit{Age of Information} (AoI).
Specifically, AoI is defined as the time that has elapsed since the freshest update at the destination was generated.
The AoI has been characterized in various queuing models, such as the single-source-single-user system with different queue disciplines in \cite{6195689,7541765} and the multiple-source system with queue management in \cite{7249268}.
Scheduling policies for AoI minimization are investigated for broadcast channels in \cite{7852321,8514816,8006590,8807257}, for multiple-access systems in \cite{8006544} and \cite{8734015}, respectively.
AoI in energy harvesting systems has been studied in \cite{7283009,8422086,8123937,feng2018age}.
When the transmssion time of updates is non-negligible or not a single time slot, it is proved in \cite{6310931} that the average AoI achieved by last-come-first-served (LCFS) with preemption discipline is lower than that of LCFS without preemption.
In \cite{8695040,8006593,8406945}, the AoI under the last-generated-first-served (LGFS) policy without and with preemption are compared. The optimal policies of preemption for average AoI minimization for a link-rate constrained status updating system is studied in \cite{8445919} and \cite{8764466}. In \cite{arafa2019timely}, the preemption policy for AoI minimization in cloud computing is investigated.

AoI as a universal metric is effective in capturing the information freshness in systems where the underlying status changes continuously in time and the corresponding time-domain structure is hard to model.
However, in many applications, the monitored signal may only change sporadically in time, e.g., for platooning in autonomous driving, vehicles are moving at a constant speed until some driving condition changes.
For such scenarios, as long as the status of the system does not change after the controller receives the update about the last status change, the information at the controller is still ``fresh''.
In other words, the information freshness should not be measured by the time that has elapsed since the generation of the latest received update.
Rather, it is related to the time that has elapsed since a status change happens at the source and information at the controller becomes outdated.
In observation of this, a metric called ``\textit{Age of Synchronization}'' (AoS) is proposed in \cite{8437927}.
It refers to the duration since the destination became \textit{desynchronized} with the source.
With this definition, the AoS in a multiple-user cache system under a given refresh rate constraint is analyzed, and a near-optimal rate allocation policy is proposed.
In \cite{8849418}, a lower bound of the time-average AoS in the broadcast network is calculated and an index based policy for AoS minimization is proposed to approximate the optimal solution to an MDP based formulation.
In the same spirit, another metric called ``Age of Incorrect Information'' (AoII) is proposed in \cite{maatouk2019age}.
AoII takes both the time that the monitor is unaware of the correct status of the system and the difference between the current estimate at the monitor and the actual state of system into the definition.
With particular penalty functions, AoII reduces to AoS.
In \cite{stamatakis2019control} and \cite{8406891}, the definition of AoI has been extended to account for the state changes of the monitored stochastic process.

In this paper, we investigate the AoS in a discrete-time single-source single-destination system under a link rate constraint. Different from \cite{8849418} and \cite{maatouk2019age}, we assume it takes multiple time slots to finish the transmission of each update and only one update can be transmitted in each slot.
The transmitter should make decisions to skip or to switch when there is a new update generated and the current transmission is unfinished yet.
We focus on the Markovian policies, and formulate the problem as a Markov decision process.
We prove the optimal policy has a multi-threshold structure, based on which we propose a structured value iteration policy to reduce the computational complexity.


\section{System Model and Problem Formulation}
\label{sys_pro}
We consider a single-link status monitoring system where a transmitter keeps sending time-stamped status updates to a monitor.
The time axis is discretized into slots.
At the beginning of each time slot, the status of the observed process may change according to an i.i.d. Bernoulli process $\{a_t\}$ with parameter $p$. Once a status change happens, a status update is generated at the source.

To simplify the analysis, as the first step, we assume that the updates are of the same size, and it takes $b$ time slots to transmit one update to the destination.
We assume the transmitter can transmit only one update at any time slot, and there is no buffer at the transmitter.
Let $w_t\in\{0,1\}$ be a binary decision variable.
If a new update arrives at the transmitter during a busy slot, the transmitter should decide either to drop the new update and keep transmitting the previous one, which is termed as \textit{skip} with $w_t=0$, or to drop the unfinished update and switch to the new one, which is termed as \textit{switch} with $w_t=1$.
We label the updates in the order of their generation times and use $T_m$ to denote the generation time of the $m$-th update.
Denote $D(t)$ as the index of the latest update received by the monitor at the beginning of the $t$-th time slot.
Then the age of synchronization is defined as
\begin{equation}
    \text{AoS}(t):=(t-T_{D(t)+1})^+,
\end{equation}
where $T_{D(t)+1}$ refers to the time when the source generates a new update after $D(t)$ and the destination becomes desynchronized, and $(x)^+=\max(x,0)$.

In order to capture the state of the system, we introduce the AoS at the transmitter as well.
Specifically, let $K(t)$ be the index of the latest update the transmitter transmits. Then the AoS at the transmitter is denoted as $(t-T_{K(t)+1})^+$.
If the transmitter {\it switches} to a new update once it is generated, the AoS at the transmitter is zero; otherwise, if it {\it skips} a new update, the transmitter becomes desynchronized with the source, and its AoS starts growing.

Let $S_i$ be the time slot that the destination has been updated successfully for the $i$-th time, where $i=0,1,2,\cdots$.
Without loss of generality, we assume $S_0=0$. $S_i$s partition the time axis into epochs, where the length of the $i$th epoch is denoted as $L_i:=S_i-S_{i-1}$.

Depending on the evolution of the monitored process, two scenarios may happen when the destination receives an updates.
For the first scenario, the update arriving at the monitor is the latest update generated by the monitored process, which is regarded as a ``fresh'' update.
Therefore, the monitor is synchronized to the observed process successfully and $\text{AoS}(S_i)=0$.
It will not increase until the monitored process changes.
On the other hand, if the monitored process changes during the transmission of the latest received update, the monitor will not be synchronized with the monitored process, thus $\text{AoS}(S_i)>0$.
For both scenarios, $\text{AoS}(t)$ will be increased by 1 for each time slot of desynchronization.
Finally, when a new update is delivered to the monitor, this epoch ends and a new one begins.
The evolution of AoS is shown in the Fig. \ref{fig:2}.

\begin{figure}[t]
    \centering
    \includegraphics[width=\columnwidth]{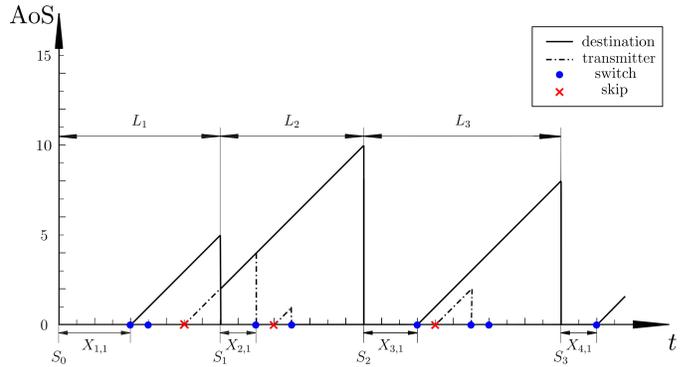}
    \vspace{-0.3in}
    \caption{Evolution of the AoS at the transmitter and at the destination with $b=4$. Once an update is received at the destination, the AoS at the destination is reset to the AoS at the transmitter.}
    \vspace{-0.2in}
    \label{fig:2}
\end{figure}

Let $X_{i,1}$ be the time between $S_i$ and the first update generation time after it.
Then, the area under the AoS curve during the $i$-th epoch, $R_i$, is given by
\begin{equation}
     R_i=\begin{cases}
     \frac{1}{2}(L_i-X_{i,1})^2,&\text{AoS}(S_{i-1})=0,\\
     \text{AoS}(S_{i-1})L_i+\frac{1}{2}L_i^2,&\text{AoS}(S_{i-1})>0.
     \end{cases}
\end{equation}

Let $M(T)$ be the number of successfully received updates over $(0,T]$.
Therefore, the cumulative AoS experienced by the monitor over $(0,T]$ can be expressed as $R(T)=\sum_{m=1}^{M(T)}R_i+r$, where
\begin{equation}
    r\triangleq\begin{cases}
    \frac{1}{2}[(T-T')^+]^2,&\text{ if }\text{AoS}(S_{M(T)})=0,\\
    \text{AoS}(S_{M(T)})\Delta_T+\frac{1}{2}\Delta_T^2,&\text{ if }\text{AoS}(S_{M(T)})>0,
    \end{cases}
\end{equation}
and $T'=S_{M(T)}+X_{M(T)+1,1},\ \Delta_T=T-S_{M(T)}$.

We consider a set of \textit{online} policies $\Pi$, in which the information available for determining $w_t$ includes the decision history $\{w_i\}_{i=1}^{t-1}$, the update generation profile $\{a_i\}_{i=1}^{t}$, as well as the generation rate $p$. Our objective is to solve the following problem:
\begin{equation}
\label{opt_pro}
    \min\limits_{\pi\in\Pi} \limsup\limits_{T\to\infty} \mathbb{E}\left[\frac{R(T)}{T}\right],
\end{equation}
where the expectation is taken with respect to the random update generation process.


\section{MDP Formulation}
\label{MDP_for}

Before formulating the MDP, we define a subset of \textit{online} policies named \textit{persistent} policies as follows,
\begin{definition}[Persistent Policy]
Under an online policy $\pi\in\Pi$, if the transmitter always keeps transmitting an update until the transmission is finished, or until the generation of a new update, i.e., it will not drop an unfinished update if no new update is generated, then this policy is a persistent policy.
\end{definition}
\if{0}
We have the following theorem.

\begin{theorem}
\label{persis}
The optimal online policy that solves (\ref{opt_pro}) is a persistent policy.
\end{theorem}

Theorem \ref{persis} can be proved through contradiction.
\fi
We can show that the optimal online policy that solves (\ref{opt_pro}) is a persistent policy. This is because if the optimal policy is not persistent, we can always construct a persistent policy to reduce the corresponding AoS. Therefore, we restrict to persistent policies in the following.

To make the optimization in (\ref{opt_pro}) tractable, we focus on Markovian policies, under which the decision only depends on the current state. The MDP is formulated as follows:

\textbf{States}:  Denote the AoS at the destination and at the transmitter, and the remaining transmission time of the unfinished update at the beginning of time slot $t$ as $d_t$, $\delta_t$ and $l_t$, respectively. Let $a_t\in\{0,1\}$ denote whether a new update is generated at the source at the beginning of time slot $t$. Then, the state of the MDP at the beginning of time slot $t$ is denoted as $\mathbf{s}_t:=(d_t,\delta_t,l_t,a_t)$. 
    
We note that $d_t=\delta_t$ if and only if $l_t=0$, which happens when the transmission of an update is finished and the destination is synchronized with the source. Otherwise, the update at the destination is more outdated than that at the transmitter, suggesting $\delta_t<d_t$. Moreover, the past transmission time of the current update $b-l_t$ should be smaller than $d_t$, since there may be multiple updates generated after the desynchronization occurs at the destination.
Besides, $b-l_t$ must be larger than $\delta_t$, since the desynchronization at the transmitter always occurs after the beginning of the current transmission. Therefore, for any valid {\it busy} state with $l_t>0$, we must have $\delta_t<b-l_t\leq d_t$, while for any valid {\it idle} state with $l_t=0$, we must have $\delta_t=d_t$.
    
 \textbf{Actions}: $w_t\in\{0,1\}$. At the beginning of time slot $t$, if $a_t=1$, when $w_t=0$, the transmitter skips the new update, and when $w_t=1$, the transmitter begins to send the new update.
    If $a_t=0$, we must have $w_t=0$.
    
 \textbf{Transition Probabilities}:
  First, we note that $a_{t+1}$ evolves according to an independent Bernoulli random variable with parameter $p$.
  Then, based on $a_t$ and $w_t$, we divided the states into two categories:
    \begin{itemize}
        \item $a_t=0$ or $w_t=0$: When there is no new update generated at the source or a new update is generated but dropped, the transmitter will continue its previous operation, i.e., either transmitting an unfinished update, or being idle. Thus, we have
        \begin{equation}
        \begin{aligned}
            &d_{t+1}=\begin{cases}
            d_t+1,&l_t\ne1\text{ and }(d_t>0\text{ or }a_t=1),\\
            \delta_t+1,&l_t=1\text{ and }\delta_t>0,\\
            0,&\text{otherwise}.
            \end{cases}\\
            &\delta_{t+1}=\begin{cases}
            \delta_t+1,&\delta_t>0\text{ or }a_t=1,\\
            0,&\text{otherwise}.
            \end{cases}\\
            &l_{t+1}=(l_t-1)^+.
        \end{aligned}
        \end{equation}
        
        \item $a_t=1, w_t=1$: If there is a new update generated at the source and the transmitter decides to switch, the transmitter will be refreshed and the destination becomes desynchronized. Thus,
        \begin{equation}
        \begin{aligned}
            &d_{t+1}=d_t+1,\quad \delta_{t+1}=0,\quad l_{t+1}=b-1.
        \end{aligned}
        \end{equation}
    \end{itemize}

 \textbf{Cost}: Let $C(\mathbf{s}_t; w_t)$ be the instantaneous AoS at the destination under state $\mathbf{s}_t$, i.e., $C(\mathbf{s}_t; w_t)=d_t$.


\section{Structural Properties of the Optimal Policy}
\label{thresh_alg}

In this section, we prove the multi-threshold structure of the optimal policy, based on which we propose a low computational complexity algorithm for value iteration to solve the MDP.
We first introduce an $\alpha$-discounted MDP as follows:
\begin{equation}
\label{alpha_dis}
    V^\alpha(\mathbf{s})=\min\limits_{w\in\mathcal{W}(\mathbf{v})}C(\mathbf{s};w)+\alpha\mathbb{E}[V^\alpha(\mathbf{s'})|\mathbf{s},w],
\end{equation}
where $0<\alpha<1$. It has been shown that the optimal policy to minimize the long-term average cost can be obtained by solving (\ref{alpha_dis}) when $\alpha\to1$.
We start with the following value iteration formulation with $V_0^\alpha(\mathbf{s})=0, \forall \mathbf{s}$:
\begin{equation}
    V_{n+1}^\alpha(\mathbf{s})=\min\limits_{w\in\mathcal{W}(\mathbf{s})}C(\mathbf{s};w)+\alpha\mathbb{E}[V_n^\alpha(\mathbf{s'})|\mathbf{s},w],
\end{equation}
where the set of allowable actions $\mathcal{W}(\mathbf{s})$ will be specified later.

In addition, we denote the state-action value functions as follows 
\begin{align}
    Q^\alpha(\mathbf{s};w)&:=C(\mathbf{s};w)+\alpha\mathbb{E}[V^\alpha(\mathbf{s}')|\mathbf{s},w],\\
    Q_n^\alpha(\mathbf{s};w)&:=C(\mathbf{s};w)+\alpha\mathbb{E}[V_n^\alpha(\mathbf{s}')|\mathbf{s},w].
\end{align}
When $a=0$, i.e., there is no new update generated in the current slot, the transmitter can only choose to continue its previous operation.
Thus,
\begin{equation}
\label{idle_q}
    V_{n+1}^\alpha(d,\delta,l,0)=Q_n^\alpha(d,\delta,l,0;0).
\end{equation}
Otherwise, when $a=1$, there is a new update generated at the current time slot, and the transmitter can choose to switch or to skip.
Thus,
\begin{equation}
\label{busy_q}
    V_{n+1}^\alpha(d,\delta,l,1)=\min\limits_{w\in\{0,1\}}Q_n(d,\delta,l,1;w).
\end{equation}

\subsection{Monotonicity of the Value Function}
\begin{lemma}
\label{expectation}
$\mathbb{E}_a[V_n^\alpha(d_1,\delta_1,l_1,a)]\leq\mathbb{E}_a[V_n^\alpha(d_2,\delta_2,l_2,a)]$ if $V_n^\alpha(d_1,\delta_1,l_1,a)\leq V_n^\alpha(d_2,\delta_2,l_2,a)$,\ for $a\in\{0,1\}$.
\end{lemma}

Lemma \ref{expectation} can be shown directly based on the definition of expectation, and it is a fundamental building block for the proofs of the remaining lemmas.

\begin{lemma}
\label{mono_d}
For any valid busy state $\mathbf{s}\in\mathcal{S}$, $V_n^\alpha(d,\delta,l,a)$ is monotonically increasing in $d$ at every iteration $n$.
\end{lemma}

\begin{lemma}
\label{mono_dd}
For any valid idle state $(d,d,0,a)$, $V_n^\alpha(d,d,0,a)$ is monotonically increasing in $d$ at every iteration $n$.
\end{lemma}

Due to space limitation, we omit the proofs of the lemmas in this paper.
With Lemmas \ref{mono_d} and \ref{mono_dd}, we can prove the monotonicity of the value function in $\delta$ as follows.
\begin{lemma}
\label{mono_s}
For any valid busy state $\mathbf{s}\in\mathcal{S}$, $V_n^\alpha(d,\delta,l,a)$ is non-decreasing in $\delta$ at every iteration $n$.
\end{lemma}
With Lemmas \ref{mono_d}-\ref{mono_s}, we obtain the following properties for the value function for busy and idle states as follows.

\begin{lemma}
\label{busy_small}
For any valid busy state $\mathbf{s}\in\mathcal{S}$ with $l>0$, $V_n^\alpha(d,\delta,l,a)\leq V_n^\alpha(d,d,0,a)$.
\end{lemma}

\begin{lemma}
\label{busy_large}
For any valid state $\mathbf{s}\in\mathcal{S}$ with $l=1$, $V_n^\alpha(\delta,\delta,0,a)\leq V_n^\alpha(d,\delta,1,a)$.
\end{lemma}

These two lemmas can be intuitively explained as follows.
Lemma \ref{busy_small} indicates that with the same AoS at the destination, the state in which the transmitter is transmitting an update is always better than the state in which the transmitter is idle.
This is because the transmission is beneficial to a successful update at the destination.
On the other hand, Lemma \ref{busy_large} suggests with the same AoS at the transmitter, the state requires one more slot to update the destination is ``worse'' than the state that the destination has just been synchronized, since the AoS at the destination in the former state is not lower than that in the latter state, and this relationship holds for any upcoming state after transition with the same action taken at the transmitter.
With Lemma \ref{busy_large}, we can prove the monotonicity of value function in $l$ as following,

\begin{lemma}
\label{mono_l}
For any valid state $\mathbf{s}\in\mathcal{S}$, $V_n^\alpha(d,\delta,l,a)$ is non-decreasing in $l$ for $l>0$ at every iteration $n$. 
\end{lemma}

Lemma \ref{mono_d}, Lemma \ref{mono_s} and Lemma \ref{mono_l} indicate that the value function has higher value with larger $d$, $\delta$ or $l$, which are consistent with our intuition that states with larger AoS at the destination or the transmitter, or longer remaining transmission time are less preferable.

\subsection{Multi-threshold Structure of the Optimal Policy}
With the monotonicity of the value function in $d$, $\delta$ and $l$ established, we are ready to obtain the multi-threshold structure of the optimal policy. Firstly, we show the existence of the thresholds on $l$ and $d$ as follows,
\begin{lemma}
\label{switch_l}
If $Q_n^\alpha(d,\delta,l,1;1)\leq Q_n^\alpha(d,\delta,l,1;0)$, then for any state $\mathbf{s'}=(d,\delta,l',1)$ with $l'>l$, if valid, we must have $Q_n^\alpha(d,\delta,l',1;1)\leq Q_n^\alpha(d,\delta,l',1;0)$.
\end{lemma}

\begin{lemma}
\label{skip_d}
If $Q_n^\alpha(d,\delta,l,1;0)\leq Q_n^\alpha(d,\delta,l,1;1)$, then for any state $\mathbf{s'}=(d',\delta,l,1)$ with $d'>d$, we must have $Q_n^\alpha(d',\delta,l,1;0)\leq Q_n^\alpha(d',\delta,l,1;1)$.
\end{lemma}

Lemma \ref{switch_l} suggests that there exists a threshold on $l$ such that the transmitter will switch to a new update only when the remaining transmission time for the unfinished update is above the threshold.
Similarly, Lemma \ref{skip_d} suggests the transmitter will skip a new update only if the current AoS at the destination is above the threshold.
\allowdisplaybreaks
Based on Lemmas \ref{switch_l} and \ref{skip_d}, we completely characterize the structural properties of the optimal policy in the following theorem.
\begin{theorem}
\label{threshold_s}
Under the optimal policy, for any fixed AoS at destination $d$ and remaining transmission time $l$, there exists a threshold $\tau_{d,l}$, such that when $\delta\geq\tau_{d,l}$, the optimal action is to transmit the new update, i.e., $w^*(d,\delta,l,1)=1$ and when $\delta<\tau_{d,l}$, the optimal action is to continue the transmitter's previous action, i.e., $w^*(d,\delta,l,1)=0$.
Especially, $\tau_{d,l}=b$ if the optimal policy for all states with $d$ and $l$ is to skip.
Besides, for any fixed $d$, the set of thresholds is decreasing in $l$, i.e., $\tau_{d,1}\geq\tau_{d,2}\geq\cdots\geq\tau_{d,b-1-\delta}$.
Similarly, for any fixed $l$, the set of threshold is increasing in $d$, i.e., $\tau_{b-l,l}\leq\tau_{b-l+1,l}\leq\cdots\leq\tau_{d,l}\leq\cdots$.
\end{theorem}
\begin{proof}
First we prove the existence of the threshold $\tau_{d,l}$. For any $s\geq0$, if the optimal policy $w^*(d,\delta,l,1)=1$, i.e., the optimal policy is to switch, we must have 
\begin{align}
    Q^\alpha(d,\delta,l,1;1)=&d+\alpha\mathbb{E}_a[V^\alpha(d+1,0,b-1,a)]\nonumber\\
    \leq&d+\alpha\mathbb{E}_a[V^\alpha(d+1,\delta+1,l-1,a)]\label{thm-1}\\
    =&Q^\alpha(d,\delta,l,1;0),\text{ for } l>1.\nonumber\\
    Q^\alpha(d,\delta,1,1;1)=&d+\alpha\mathbb{E}_a[V^\alpha(d+1,0,b-1,a)]\nonumber\\
    \leq&d+\alpha\mathbb{E}_a[V^\alpha(\delta+1,\delta+1,0,a)]\label{thm-2}\\
    =&Q^\alpha(d,s,1,1;0).\nonumber
\end{align}
Then, for any $\delta'>\delta$, 
\begin{align}
    Q^\alpha(d,\delta',l,1;1)=&d+\alpha\mathbb{E}_a[V^\alpha(d+1,0,b-1,a)]\nonumber\\
    \leq&d+\alpha\mathbb{E}_a[V^\alpha(d+1,\delta+1,l-1,a)]\label{thm-3}\\
    \leq&d+\alpha\mathbb{E}_a[V^\alpha(d+1,\delta'+1,l-1,a)]\hspace{-0.1in}\label{thm-4}\\
    =&Q^\alpha(d,\delta',l,1;0),\text{ for } l>1,\nonumber\\
    Q^\alpha(d,\delta',1,1;1)=&d+\alpha\mathbb{E}_a[V^\alpha(d+1,0,b-1,a)]\nonumber\\
    \leq&d+\alpha\mathbb{E}_a[V^\alpha(\delta+1,\delta+1,0,a)]\label{thm-5}\\
    \leq&d+\alpha\mathbb{E}_a[V^\alpha(\delta'+1,\delta'+1,0,a)]\label{thm-6}\\
    =&Q^\alpha(d,\delta',1,1;0),\nonumber
\end{align}
where (\ref{thm-3}) and (\ref{thm-5}) are due to (\ref{thm-1}) and (\ref{thm-2}) respectively, and (\ref{thm-4})(\ref{thm-6}) are based on Lemma \ref{mono_s}.
Thus, for $\delta'>\delta$, the optimal policy for state $(d,\delta',l,1)$ is to switch. 

Similarly, for any $\delta>0$, if the optimal policy $w^*(d,\delta,l,1)=0$, i.e., the optimal policy is to skip, we have
\begin{align}
    Q^\alpha(d,\delta,l,1;0)=&d+\alpha\mathbb{E}_a[V^\alpha(d+1,\delta+1,l-1,a)]\nonumber\\
    \leq&d+\alpha\mathbb{E}_a[V^\alpha(d+1,0,b-1,a)]\label{thm-7}\\
    =&Q^\alpha(d,\delta,l,1;1),\text{ for } l>1.\nonumber\\
    Q^\alpha(d,\delta,1,1;0)=&d+\alpha\mathbb{E}_a[V^\alpha(\delta+1,\delta+1,0,a)]\nonumber\\
    \leq&d+\alpha\mathbb{E}_a[V^\alpha(d+1,0,b-1,a)]\label{thm-8}\\
    =&Q^\alpha(d,\delta,1,1;1).\nonumber
\end{align}
Then, for any $0<\delta'<\delta$, we have
\begin{align}
    Q^\alpha(d,\delta',l,1;0)=&d+\alpha\mathbb{E}_a[V^\alpha(d+1,\delta'+1,l-1,a)]\nonumber\\
    \leq&d+\alpha\mathbb{E}_a[V^\alpha(d+1,\delta+1,l-1,a)]\label{thm-9}\\
    \leq&d+\alpha\mathbb{E}_a[V^\alpha(d+1,0,b-1,a)]\label{thm-10}\\
    =&Q^\alpha(d,\delta',l,1;1),\text{ for } l>1,\nonumber\\
    Q^\alpha(d,\delta',l,1;0)=&d+\alpha\mathbb{E}_a[V^\alpha(\delta'+1,\delta'+1,0,a)]\nonumber\\
    \leq&d+\alpha\mathbb{E}_a[V^\alpha(\delta+1,\delta+1,0,a)]\label{thm-11}\\
    \leq&d+\alpha\mathbb{E}_a[V^\alpha(d+1,0,b-1,a)]\label{thm-12}\\
    =&Q^\alpha(d,\delta',1,1;1),\nonumber
\end{align}
where (\ref{thm-9}) is due to Lemma \ref{mono_s}, (\ref{thm-10}) is based on (\ref{thm-7}), (\ref{thm-11}) follows from Lemma \ref{mono_dd}, and (\ref{thm-12}) is due to (\ref{thm-8}).

Following similar argument, for $\delta'=0$, we have 
\begin{align*}
    Q^\alpha(d,0,l,1;0)=&d+\alpha\mathbb{E}_a[V^\alpha(d+1,1,l-1,a)]\\
    \leq&d+\alpha\mathbb{E}_a[V^\alpha(d+1,\delta+1,l-1,a)]\\
    \leq&d+\alpha\mathbb{E}_a[V^\alpha(d+1,0,b-1,a)]\\
    =&Q^\alpha(d,0,l,1;1),\text{ for } l>1.\\
    Q^\alpha(d,0,l,1;0)=&d+\alpha\mathbb{E}_a[V^\alpha(1,1,0,a)]\\
    \leq&d+\alpha\mathbb{E}_a[V^\alpha(\delta+1,\delta+1,0,a)]\\
    \leq&d+\alpha\mathbb{E}_a[V^\alpha(d+1,0,b-1,a)]\\
    =&Q^\alpha(d,0,1,1;1).
\end{align*}

Combining both cases, for any $\delta'<\delta$, the optimal policy for state $(d,\delta',l,1)$ is to skip. 

Thus, there exists a threshold $\tau_{d,l}$ for states with fixed $d$ and $l$, such that when $\delta\geq\tau_{d,l}$, the optimal action $w^*(d,\delta,l,1)=1$ and when $\delta<\tau_{d,l}$, $w^*(d,\delta,l,1)=0$.

Then we prove the monotonicity of $\tau_{d,l}$ in $l$.
Consider the case when $l>1$ first. The definition of $\tau_{d,l}$ indicates that 
\begin{equation}
\begin{aligned}
     Q^\alpha(d,\tau_{d,l},l,1;1)\leq Q^\alpha(d,\tau_{d,l},l,1;0).
\end{aligned}
\end{equation}
Then, for $l'>l$, if the state $\mathbf{s}=(d,\delta,l,1)$ is valid, according to Lemma \ref{switch_l}, 
\begin{equation}
\begin{aligned}
    Q^\alpha(d,\tau_{d,l},l',1;1)\leq Q^\alpha(d,\tau_{d,l},l',1;0),
\end{aligned}
\end{equation}
which suggests that $\tau_{d,l'}\leq\tau_{d,l}$.
Thus, $\tau_{d,l}$ is decreasing in $l$.

Finally, we prove the monotonicity of $\tau_{d,l}$ in $d$.
For any fixed $d,l$, if $\tau_{d,l}=0$, i.e., the optimal policy is to switch, we have $\tau_{d+1,l}\geq\tau_{d,l}=0$ since the minimum value of the threshold is non-negative.
Otherwise, state $(d,\tau_{d,l}-1,l,1)$ is valid and its optimal policy is to skip.
Therefore, according to Lemma \ref{skip_d}, the optimal policy of state $(d+1,\tau_{d,l}-1,l,1)$ is to skip as well, which suggests $\tau_{d+1,l}\geq\tau_{d,l}$.
Combining two cases, the monotonicity of $\tau_{d,l}$ in $d$ is established.
\end{proof}

\subsection{Structured Value Iteration}
To reduce the computational complexity, we leverage the multi-threshold structure during the value iteration procedure, similar to the structured value iteration algorithm in \cite{8764466}. We omit the detailed algorithm for the brevity of the paper.


\section{Numerical Results}
\label{num_res}

\subsection{The Multi-threshold Policy}
Since the number of states in the original MDP is infinite, numerical iteration over all states is impractical.
Therefore, we propose an approximate MDP as follows: defining the largest AoS at the destination as $d_{max}$ and truncating the state space of the original MDP as $\mathcal{S}_m=\{\mathbf{s}\in\mathcal{S}:d\leq d_{max}\}$.
Then, we set $\mathbf{s}_{t+1}$ as the corresponding capped state if it is outside $\mathcal{S}_m$.
It is shown that the approximate MDP is identical to the original MDP when $d_{max}\to\infty$ \cite{sennott1997computing}.
Thus, when implementing the structured value iteration, we set $b=10$, $d_{max}=400$ and $\alpha=0.9999$.

Fig. \ref{threhold_result} shows the thresholds on $\delta$ with fixed $d$ and $l$.
For any state $\mathbf{s}=(d,\delta,l,1)$, the optimal action is to switch if $\delta$ is above the bar located at the corresponding $d$ and $l$, otherwise the optimal action is to skip.
We note that the thresholds are monotonically increasing in $d$ and decreasing in $l$, as predicted by Theorem \ref{threshold_s}.
Besides, as the update generation rate $p$ increases, the optimal policy is more inclined to skip at the same state.

\subsection{Performance Comparison}
We evaluate the average AoS under the optimal policy, the AoI-optimal policy and two baseline policies, \textit{always skip} and \textit{always switch}, over $10^7$ time slots.
Under the \textit{always skip} policy, the transmitter will always drop the new update if there is a update being transmitted, while under the \textit{always switch} policy, the transmitter will always drop the update being transmission and switch to the new update.
Besides, in order to examine the difference between AoS and AoI, we also study the AoI performance under those policies as well as the AoI-optimal policy in \cite{8764466} under the same setting. The result is shown in Fig. \ref{result}.

We notice that when the generation rate $p\to1$, the average AoS under the \textit{always switch} policy becomes unbounded, while those under the AoS-optimal policy and the \textit{always skip} policy tend to be identical.
The results can be intuitively explained as follows: when the updates are generated at the source frequently, the greedy policy, which prefers to finish the current transmission and decrease the AoS at the destination as soon as possible, will be optimal.
The optimal policy thus behaves the same as the \textit{always skip} policy with high probability.
On the other hand, when the transmitter always switches to new updates, it will not be able to finish the transmission of any update over long periods of time, leading to constantly growing AoS at the destination.

We also notice that in Fig. \ref{result}, when the generation rate $p$ is small, the average AoS under the AoS-optimal policy and the \textit{always switch} policy tend to be the same, which is lower than that under the \textit{always skip} policy.
This is because when $p$ is small, the chance that a new update is generated when the transmitter is busy is small.
Thus, the AoS-optimal policy behaves similarly to the \textit{always switch} policy.

As for the AoI-optimal policy, when the generation rate $p\to1$, the average AoS under it tends to be identical to those under the AoS-optimal policy and \textit{always skip} policy.
But when the generation rate $p$ is small, the average AoS under it is larger than those under the AoS-optimal policy and \textit{always switch} policy.
Thus, the optimal policy for AoI minimization is not efficient for AoS minimization.

For the average AoI performance, all policies, including the AoI-optimal policy, the AoS-optimal policy and the two baseline policies, perform closely when $p$ is small.
When $p\to1$, all policies except \textit{always switch} perform similarly.

Perhaps the most interesting distinction between AoS and AoI lies in the different trending curves as the update generation rate $p$ increases: the minimum AoI monotonically decreases as $p$ increases, while the minimum AoS exhibits the opposite trend.
This is because AoI only depends on the age of the freshest information at the destination without considering the underlying status evolution.
Thus, any information ages linearly in time since its generation.
Correspondingly, lower generation rate increases the duration between two successful updates at the destination, leading to higher AoI.
AoS, on the other hand, depends on the ``change'' of the system status.
Thus, lower generation rate implies that each update can stay fresh for a longer time, and the AoS is actually lower.

\begin{figure}[t]
\centering
\vspace{-0.1in}
\includegraphics[width=\columnwidth]{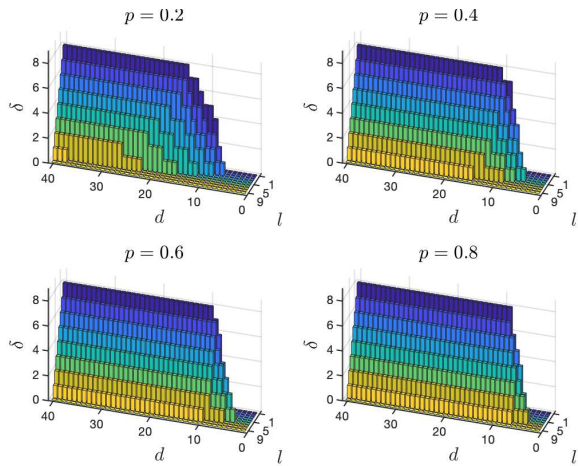}
\vspace{-0.4in}
\caption{Thresholds on $\delta$ with different generation rate $p$.}
\vspace{-0.1in}
\label{threhold_result}
\end{figure}

\begin{figure}[t]
\centering
\vspace{-0.1in}
\includegraphics[width=3.5in]{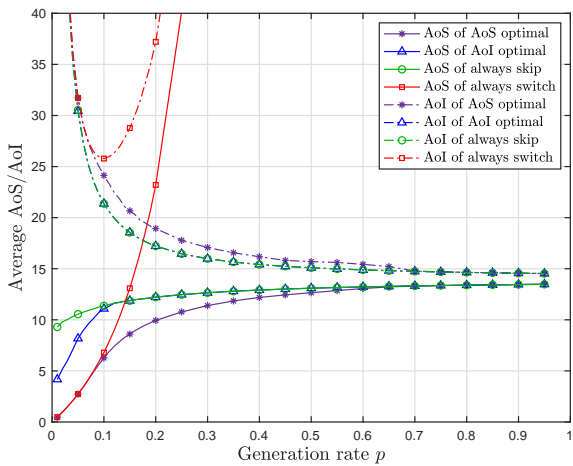}
\vspace{-0.3in}
\caption{Average AoS and AoI with different generation rate $p$.}
\vspace{-0.2in}
\label{result}
\end{figure}

\section{Conclusions}
\label{conclu}
In this paper, we investigate timely synchronization in a status monitoring system with occasional status changes at the source.
We adopt the metric AoS and formulate the problem as an MDP.
Theoretical analysis shows the optimal policy has a multi-threshold structure. Numerical results corroborate the theoretical results.


\bibliography{biblio_icc}

\end{document}